\documentclass[11pt]{amsart}

\pdfoutput=1

\usepackage[text={420pt,660pt},centering]{geometry}

\usepackage{cancel}
\usepackage{esint,amssymb}
\usepackage{graphicx}
\usepackage{mathtools} 
\usepackage[colorlinks=true, pdfstartview=FitV, linkcolor=blue, citecolor=blue, urlcolor=blue,pagebackref=false]{hyperref}
\usepackage{orcidlink}
\usepackage{microtype}

\usepackage{bm}
\usepackage{dsfont}
\usepackage{mathrsfs}
\usepackage{xcolor}
\usepackage{accents} \usepackage{enumitem} \parskip= 2pt
\usepackage[normalem]{ulem}

\newtheorem{proposition}{Proposition}
\newtheorem{theorem}[proposition]{Theorem}
\newtheorem{lemma}[proposition]{Lemma}

\theoremstyle{remark}
\newtheorem{remark}[proposition]{Remark}

\theoremstyle{definition}
\newtheorem{definition}[proposition]{Definition}

\numberwithin{equation}{section}
\numberwithin{proposition}{section}
\numberwithin{figure}{section}
\numberwithin{table}{section}

\newcommand{\N}{\mathbb{N}}

\newcommand{\R}{\mathbb{R}}

\newcommand{\E}{\mathbb{E}}

\renewcommand{\S}{\mathbf{S}}

\newcommand{\eps}{\varepsilon}

\renewcommand{\leq}{\leqslant}
\renewcommand{\geq}{\geqslant}

\newcommand{\Ll}{\left}
\newcommand{\Rr}{\right}
\renewcommand{\d}{\mathrm{d}}

\newcommand{\D}{q}

\newcommand{\la}{\left\langle}
\newcommand{\ra}{\right\rangle}

\newcommand{\diag}{\mathrm{diag}}

\newcommand{\one}{{\boldsymbol{1}}}
\newcommand{\Sym}{{\mathrm{Sym}(\D)}}

\newcommand{\sP}{\mathscr{P}}

\newcommand{\cc}{{\mathrm{c}}}
\newcommand{\cw}{{\mathrm{CW}}}
\newcommand{\bfs}{{\mathbf{s}}}

\definecolor{darkgreen}{rgb}{0,0.5,0}
\definecolor{darkblue}{rgb}{0,0,0.7}

\newcommand{\rv}[1]{\textcolor{black}{#1}}

\begin{document}

\author[Hong-Bin Chen]{Hong-Bin Chen}
\address[Hong-Bin Chen]{Institut des Hautes Études Scientifiques, Bures-sur-Yvette, France}
\email{\href{mailto: hbchen@ihes.fr}{ hbchen@ihes.fr}}


\title[Color symmetry and ferromagnetism in Potts spin glass]{Color symmetry and ferromagnetism\\ in Potts spin glass}

\begin{abstract}
We consider the Potts spin glass with additional ferromagnetic interaction parametrized by $t$. It has long been observed that the Potts color symmetry breaking for the spin glass order parameter is closely related to the ferromagnetic phase transition. To clarify this, we identify a single critical value $t_\mathrm{c}$, which marks the onset of both color symmetry breaking and the transition to ferromagnetism. 
\end{abstract}

\maketitle

\section{Introduction}

\subsection{Setting}

Throughout, for matrices or vectors $a$ and $b$ of the same dimension, we denote by $a\cdot b$ the entry-wise inner product and write $|a|=\sqrt{a\cdot a}$. 
We write $\R_+=[0,\infty)$ and $\R_{++} = (0,\infty)$.

Fix an integer $\D\geq 2$ and let 
$\mu$ be the uniform probability measure on the standard basis
$\{e_1,\dots,e_\D\}$ of $\R^\D$ viewed as the state space of the $\D$-Potts spins. 
We interpret $\{1,\dots,\D\}$ as the labels of $\D$ instinct colors.
For each $N\in\N$, we sample column vectors $\sigma_{\bullet i}=(\sigma_{di})_{1\leq d\leq \D}$ independently from $\mu$, for each $i\in\{1,\dots,N\}$. We denote by $\sigma = (\sigma_{di})_{1\leq d\leq \D,\, 1\leq i\leq N}$ the spin configuration, which is viewed as a $\D\times N$ matrix.

The $\R^\D$-valued \textit{mean magnetization} is defined as
\begin{align}\label{e.m_N=}
    m_N = N^{-1} \sum_{i=1}^N \sigma_{\bullet i}.
\end{align}
Notice that for Potts spins, we have the following relations
\begin{align}\label{e.self-overlap=m_N}
    N^{-1}\sigma\sigma^\intercal = \diag\Ll(m_N\Rr),\qquad
    m_N= N^{-1}\Big(\#\Ll\{i\in\{1,\dots,N\}:\: \sigma_{\bullet i}=e_d\Rr\}\Big)_{1\leq d\leq \D}.
\end{align}
Hence, the $\R^{\D\times\D}$-valued \textit{self-overlap} $N^{-1}\sigma\sigma^\intercal$ contains the same information of the mean magnetization $m_N$.

We consider the spin glass interaction
\begin{align*}
    H_N(\sigma) = \frac{1}{\sqrt{N}} \sum_{i,j=1}^N g_{ij}\sigma_{\bullet i}\cdot \sigma_{\bullet j}.
\end{align*}
where $(g_{ij})_{1\leq i,j\leq N}$ is a collection of i.i.d.\ standard Gaussian random variables.
Throughout, we fix $\beta\geq 0$ to be the inverse temperature.
For $N\in\N$, $t\in \R_+$, and $x\in \R^\D$, we consider the augmented Hamiltonian
\begin{align}\label{e.H_beta,N=}
    H_{\beta,N}(t,x,\sigma) = \beta H_N(\sigma) + \Ll(t- \frac{\beta^2}{2}\Rr)N|m_N|^2 + N x\cdot m_N
\end{align}
where on the right-hand side the second term accounts for the additional ferromagnetic interaction and the third term is an external magnetic field.
At $t=0$, we can use~\eqref{e.self-overlap=m_N} to write
\begin{align}\label{e.H(0,0,sigma)=}
    H_{\beta,N}(0,x,\sigma) = \beta H_N(\sigma) - \frac{\beta^2}{2}N\Ll|\frac{\sigma\sigma^\intercal}{N}\Rr|^2 + N x\cdot m_N
\end{align}
which is the Hamiltonian in the model with \textit{self-overlap correction} considered in~\cite{chen2023onparisi}. The self-overlap correction refers to $- \frac{\beta^2}{2}N\Ll|\frac{\sigma\sigma^\intercal}{N}\Rr|^2$ which is $-\frac{1}{2}$ times the variance of $H_N(\sigma)$ (viewed as a Gaussian random variable). The main result in~\cite{chen2023onparisi} states that the Parisi formula for the model~\eqref{e.H(0,0,sigma)=} at $x=0$ is an infimum taken over paths respecting the color symmetry of the Potts spin (i.e.\ invariant under permutation of color labels $\{1,\dots ,\D\}$).
In the setting of vector spin glass which includes the Potts one as a special case, the model with self-overlap correction~\cite{chen2023self,chen2023on} is easier to deal with than the standard ones, partly due to the reason that the Parisi formula is an infimum rather than a $\sup\inf$ formula~\cite{pan.potts,pan.vec} for the standard model without the self-overlap correction. This is the reason for us to view~\eqref{e.H(0,0,sigma)=} as the base model and add ferromagnetic interactions on top of it as in~\eqref{e.H_beta,N=}.

We also mention that it is natural to view the term $\Ll(t- \frac{\beta^2}{2}\Rr)N|m_N|^2$ through~\eqref{e.m_N=} as shifting the mean of $(g_{ij})$, which is a more common perspective used in, e.g., \cite{elderfield1983curious,elderfield1983novel,elderfield1983spin,lage1983mixed,gross1985mean,de1995static,caltagirone2012dynamical}.

For $N\in\N$ and $(t,x) \in \R_+\times \R^\D$, the free energy is defined as 
\begin{align}\label{e.F_N=}
    F_{\beta,N}(t,x) = \frac{1}{N}\E \log\int \exp\Ll(H_{\beta,N}(t,x,\sigma)\Rr)\otimes_{i=1}^N \d \mu(\sigma_{\bullet i})
\end{align}
where $\E$ averages over $(g_{ij})$.
The associated Gibbs measure is denoted as
\begin{align}\label{e.gibbs}
    \la\cdot\ra_{t,x} \propto \exp\Ll(H_{\beta,N}(t,x,\sigma)\Rr)\otimes_{i=1}^N \d \mu(\sigma_{\bullet i})
\end{align}
where we choose to omit $\beta$ and $N$ from the notation since $\beta$ is fixed and $N$ will be clear from the context.

It has long been observed in physics~\cite{elderfield1983curious,elderfield1983novel,elderfield1983spin,lage1983mixed,gross1985mean,de1995static,caltagirone2012dynamical} that sometimes (i.e.\ for certain $\D$ and $\beta$) additional antiferromagnetic interaction (i.e.\ $t<\frac{\beta^2}{2}$) is needed to ensure that the spin glass order parameter respects the Potts color symmetry. This suggests a close link between ferromagnetism and color symmetry breaking. But to the best knowledge of the author, there has not been rigorous studies of it yet.

To present the main result that clarifies this link, we need to describe the Parisi formula for the limit of free energy.

\subsection{Parisi formulae}

We view $F_{\beta,N}$ in~\eqref{e.F_N=} as a function on $\R_+\times \R^\D$. We denote the limit of initial condition as
\begin{align}\label{e.sP(x)=}
    \sP_\beta(x ) = \lim_{N\to\infty} F_{\beta,N}(0,x),\quad\forall x \in \R^\D.
\end{align}
Recall the discussion below~\eqref{e.H(0,0,sigma)=} that the case $t=0$ corresponds to the model with self-overlap correction. The existence of the limit in~\eqref{e.sP(x)=} is ensured by~\cite[Theorem~1.1]{chen2023self}.
One should think of $\sP_\beta(x)$ as the Parisi formula associated with the model $H_{\beta,N}(0,x,\sigma)$.

In contrast to the sup-inf formula derived by Panchenko in~\cite{pan.potts}, the expression for $\sP_\beta(x)$ presented here involves only an infimum, owing to the inclusion of the self-overlap correction term.
Indeed, \cite[Theorem~1.1]{chen2023self} gives that, for each $x \in \R^\D$, 
\begin{align}\label{e.Parisi_formula}
    \sP_\beta(x) = \inf_{\pi \in \Pi} \sP_\beta(\pi, x).
\end{align}
We now describe the right-hand side of this expression. Let $\S^\D_+$ denote the set of $\D \times \D$ real symmetric positive semi-definite matrices. The collection $\Pi$ consists of left-continuous paths $\pi:[0,1] \to \S^\D_+$ that are continuous at $1$ and increasing in the sense that $\pi(s') - \pi(s) \in \S^\D_+$ whenever $s' \geq s$. The Parisi functional $\sP_\beta(\pi, x)$ is defined explicitly in~\cite[(1.5)]{chen2023self} (see also~\cite[(1.7)]{chen2023on}); we omit the full expression, as it is not needed here.

For each $z\in \R^\D$, we consider the subcollection $\Pi(z) = \Ll\{\pi\in\Pi:\: \pi(1)=\diag(z)\Rr\}$.
By~\cite[Theorem~1.1 (3)]{chen2023on} (see \eqref{e.diag(nabla_sP)} in Remark~\ref{r.compare_note} to match the notation), we can refine~\eqref{e.Parisi_formula} into
\begin{align*}
    \sP_\beta(x) = \inf_{\pi\in\Pi\Ll(\nabla\sP_\beta(x)\Rr)} \sP_\beta(\pi, x).
\end{align*}
In other words, the Parisi formula in~\eqref{e.Parisi_formula} optimizes over paths with a fixed endpoint specified by $\nabla \sP_\beta(x)$.
Using this, we can rewrite the Hopf--Lax formula to be given in~\eqref{e.hopf-lax} as
\begin{align}\label{e.f(t,0)=hopf_lax_path}
    \lim_{N\to\infty} F_{\beta,N}(t,0) = \sup_{y\in\R^\D} \inf_{\pi\in\Pi\Ll(\nabla\sP_\beta(y)\Rr)}\Ll\{\sP_\beta(\pi, y)-\frac{|y|^2}{4t}\Rr\}, \quad\forall t\in\R_+.
\end{align}
Recall that $F_{\beta,N}(t,0)$ denotes the free energy in the absence of an external field, i.e., when $x = 0$. Using the Parisi formula~\eqref{e.f(t,0)=hopf_lax_path}, we can now define the notion of color symmetry.

\subsection{Color symmetry}

Let $\Sym$ be the group consisting of permutations on $\{1,\dots,\D\}$ (labels of colors). For every $\bfs \in \Sym$ and $x =(x_1,\dots,x_\D) \in\R^\D$, we write
\begin{align}\label{e.x^s=}
    x^\bfs  = \Ll( x_{\bfs(1)},\dots ,x_{\bfs(\D)}\Rr).
\end{align}
In this notation, any $\bfs\in\Sym$ acts on the spin configuration $\sigma$ through
\begin{align}\label{e.permute_spin}
    \sigma^\bfs = (\sigma_{\bullet i}^\bfs)_{1\leq i\leq N}
\end{align}
where we apply the same permutation $\bfs$ to every single spin (column vector).

Let $\mathcal Z$ denote the set of increasing, left-continuous paths $\zeta:[0,1] \to [0,1]$ satisfying $\zeta(1) := \lim_{s \uparrow 1} \zeta(s) = 1$. 
Let $\mathbf{Id}$ be the $\D\times\D$ identity matrix and let $\one$ be the vector in $\R^\D$ with all entries equal to $1$. 
We define
\begin{align}\label{e.Psi=}
    \Psi(s) = \frac{s}{\D} \mathbf{Id} + \frac{1 - s}{\D^2} \one \one^\intercal ,\quad \forall s \in [0,1].
\end{align}
These objects will be used in the discussion below.

We now give a heuristic interpretation of the parameters appearing in~\eqref{e.f(t,0)=hopf_lax_path}. Let $y$ and $\pi$ denote optimal values. Intuitively, the asymptotic distribution of the overlap $N^{-1} \sigma \sigma'^\intercal$ under $\E\langle \cdot \rangle_{t,0}$ (where $\sigma$ and $\sigma'$ are independent samples from $\langle \cdot \rangle_{t,0}$) should correspond to the law of $\pi(U)$, where $U$ is uniformly distributed on $[0,1]$.

Similarly, the limiting self-overlap $N^{-1} \sigma \sigma^\intercal$ should correspond to the endpoint $\pi(1)$, which equals $\diag\big(\nabla \sP_\beta(y)\big)$. While this correspondence is not exact, we adopt this heuristic viewpoint to guide intuition.

Color symmetry suggests that the optimal path $\pi$ should be invariant under permutations of spin labels. More precisely,
\begin{align*}
    \big( \pi_{dd'}(s) \big)_{1 \leq d, d' \leq \D} = \big( \pi_{\bfs(d)\bfs(d')}(s) \big)_{1 \leq d, d' \leq \D}, \quad \forall s \in [0,1), \quad \forall \bfs \in \Sym.
\end{align*}
Moreover, since the entries of $N^{-1} \sigma \sigma'^\intercal$ sum to $1$, we expect that $\sum_{d,d'=1}^\D \pi_{dd'} = 1$. Under these constraints, any such $\pi$ must take the form $\pi = \Psi \circ \zeta$ for some $\zeta \in \mathcal Z$.

In particular, the endpoint satisfies $\pi(1) = \D^{-1} \mathbf{Id}$, implying that $\nabla \sP_\beta(y) = \D^{-1} \one$. Although the argument is non-rigorous, it serves as motivation for the definition that follows.

\begin{definition}\label{d.color_symmetry}
Let $t \in \R_+$ and consider the case $x = 0$ as in~\eqref{e.F_N=}. We say that \emph{color symmetry is preserved} if there exists a maximizer $y \in \R^\D$ of the supremum in~\eqref{e.f(t,0)=hopf_lax_path} such that $\nabla \sP_\beta(y) = \D^{-1} \one$, and the infimum in~\eqref{e.f(t,0)=hopf_lax_path} is attained over symmetric paths, specifically:
\begin{align}\label{e.colorsymparisi}
    \lim_{N \to \infty} F_{\beta,N}(t,0) = \inf_{\zeta \in \mathcal Z} \sP_\beta(\Psi \circ \zeta, y) - \frac{|y|^2}{4t},
\end{align}
where $\Psi$ is defined in~\eqref{e.Psi=} and $\mathcal Z$ is introduced above~\eqref{e.Psi=}. 

If no such $y$ and path structure exist, we say that \emph{color symmetry is broken}.
\end{definition}

\subsection{Main result}
Recall the mean magnetization $m_N$ from~\eqref{e.m_N=} and the Gibbs measure from~\eqref{e.gibbs}.

\begin{theorem}
\label{thm.main}
Fix $\D \in \N$ and $\beta > 0$. Consider the zero external field case, i.e., $x = 0$. There exists a critical value $t_\cc > 0$, defined as the first point where the map $t \mapsto \lim_{N \to \infty} F_{\beta,N}(t,0)$ becomes non-linear (see~\eqref{e.t_c=}), such that the following phase transition occurs:

For $t < t_\cc$:
    \begin{itemize}
        \item Color symmetry is preserved (which also holds at $t=t_\cc$ by continuity).
        \item The mean magnetization (after proper re-centering) is zero: for any sequence $(x_N)_{N\in\N}$ converging to $0$, we have 
    \begin{align*}
        \lim_{N\to\infty} \E \la \Ll|m_N-\D^{-1}\one\Rr|\ra_{t,x_N} =0.
    \end{align*}
    \end{itemize}

For $t > t_\cc$:
    \begin{itemize}
        \item Color symmetry is broken.
        \item Spontaneous magnetization appears: there is a strictly increasing sequence $(N_n)_{n\in\N}$ of positive integers and $(x_n)_{n\in\N}$ converging to $0$ such that $\lim_{n\to\infty} \E \la m_{N_n}\ra_{t,x_n}$ exists but is not $\D^{-1}\one$.
    \end{itemize}

The critical value $t_\cc$ also corresponds to the point where the variational representation in~\eqref{e.f(t,0)=hopf_lax_path} of $\lim_{N \to \infty} F_{\beta,N}(t,0)$ transitions from having a unique maximizer to multiple maximizers.
\end{theorem}

It is natural to expect that ferromagnetism implies color symmetry breaking, as the ferromagnetic phase favors one color over the others. However, it is noteworthy that these two phase transitions occur simultaneously, despite being associated with distinct types of order parameters.
We make two comments on the model at $x=0$:
\begin{itemize}
    \item The law of the overlap $N^{-1}\sigma\sigma'^\intercal$ is color-symmetric (invariant under the action in~\eqref{e.permute_spin}) for all $t$, and therefore so is its limit (if exists). However, for $t > t_\cc$, due to the breaking of color symmetry, an optimal path $\pi$ in the Parisi formula~\eqref{e.f(t,0)=hopf_lax_path} is no longer color-symmetric. Consequently, in this supercritical regime, $N^{-1}\sigma\sigma'^\intercal$ does not converge to $\pi$. An analogous statement applies to the self-overlap when $t > t_\cc$.
    
    \item Since $t_\cc > 0$, color symmetry is always preserved at $t = 0$ for any $\beta$. In this case, the Hamiltonian is given by~\eqref{e.H(0,0,sigma)=}, where a self-overlap correction is included. This correction can be interpreted as subtracting from $\beta H_N(\sigma)$ the appropriate amount of ferromagnetic interaction. Consequently, at $t = 0$, the Parisi formula~\eqref{e.sP(x)=}--\eqref{e.Parisi_formula} takes the form of an infimum, as in the Sherrington--Kirkpatrick model, whereas for all $t > 0$, it assumes a sup-inf structure, as in~\eqref{e.f(t,0)=hopf_lax_path}.
    
    \rv{For $t \in (0, t_\cc]$, the color symmetry (as defined in Definition~\ref{d.color_symmetry}) allows the variational formula~\eqref{e.f(t,0)=hopf_lax_path} to reduce to an infimum representation given in~\eqref{e.colorsymparisi}. In contrast, for $t > t_\cc$, multiple maximizers of~\eqref{e.f(t,0)=hopf_lax_path} exist (see Proposition~\ref{p.t_c}~\eqref{i.p.t_c_3}), and their structure is too implicit to permit a similar simplification.
}
\end{itemize}

\subsection{Related works}
In addition to the aforementioned sources of motivation~\cite{elderfield1983curious,elderfield1983novel,elderfield1983spin,lage1983mixed,gross1985mean,de1995static,caltagirone2012dynamical}, we highlight another motivation stemming from~\cite{mourrat2024color}.
When $t=\frac{\beta^2}{2}$ and $x=0$, the model becomes standard. Hence, it is interesting to compare $t_\cc$ with $\frac{\beta^2}{2}$.
We display the dependence of $t_\cc$ on $\D$ and $\beta$ by writing $t_\cc=t_\cc(\D,\beta)$.
Recently, \cite{mourrat2024color} showed that for sufficiently large $\D$, there is a range of $\beta$ such that $t_\cc(\D,\beta) <\frac{\beta^2}{2}$ (color symmetry breaking in the standard Potts spin glass). 
In this work, however, we do not pursue questions along this line.

Lastly, we briefly mention other related works.
When $\D=2$, the Potts spin glass is essentially the same as the Sherrington--Kirkpatrick model~\cite{sherrington1975solvable,parisi79,parisi80,gue03,Tpaper,pan} (but we do not see the phenomenon discussed here, since the self-overlap of Ising spins is constantly one). The Potts spin glass was introduced in~\cite{elderfield1983curious} and the associated Parisi formula was first established in~\cite{pan.potts}. Recently,~\cite{bates2023parisi} showed that the spin-glass order parameter is color symmetric in the constrained model, which inspired \cite{chen2023onparisi} to show a similar result for the model with self-overlap correction. The argument in the two works was further generalized in~\cite{issa2024existence}. The studies of the non-disordered Curie--Weiss--Potts model include~\cite{ellis1992limit, costeniuc2005complete, eichelsbacher2015rates, lee2022energy}.
The Potts spin glass model in consideration here has a conventional order parameter, the mean magnetization, in addition to the spin glass order parameter. 
Studies of the interplay of the two kinds of parameters include~\cite{mottishaw1986first,chen2014mixed,camilli2022inference, Baldwin2023, chen2024free, baik2017fluctuations, banerjee2020fluctuation}. Spin glass models with self-overlap correction have been considered in~\cite{chen2023self,chen2023on}, inspired by the Hamilton--Jacobi equation approach to spin glass. Also, at the core of the analysis is the Hopf--Lax formula for the solution of such an equation. Considerations along this line include~\cite{mourrat2022parisi,mourrat2020extending,mourrat2021nonconvex,mourrat2023free,HJbook,HJcritical,guerra2001sum,barra1,barra2,abarra,barramulti,genovese2009mechanical,HB1,HBJ,chen2022statistical,chen2023free}.

\section{Proofs}
We prove the main result as described in the introduction. 
\rv{Since all proofs are presented in this section, we provide the following outline:
\begin{itemize}
    \item Section~\ref{s.2.1} establishes the analytic properties of $\sP_\beta$ (Lemma~\ref{l.sP_property}). 
    \item Section~\ref{s.2.2} recalls that for $t > 0$, $\lim_{N \to \infty} F_{\beta,N}(t,x)$ is given by a Hopf--Lax formula (Proposition~\ref{p.cvg_F_N}). We also collect useful properties of this limit in Lemmas~\ref{l.envelop} and~\ref{l.color_sym_max_in_R1}.
    \item Section~\ref{s.2.3} introduces $t_\cc$, defined explicitly in~\eqref{e.t_c=}, and proves its basic properties, particularly concerning the differentiability of $\lim_{N \to \infty} F_{\beta,N}(0,\cdot)$ (Proposition~\ref{p.t_c}).
    \item Section~\ref{s.2.4.color_sym_break} proves the breaking of color symmetry for $t > t_\cc$ (Proposition~\ref{p.color_symmetry}).
    \item Section~\ref{s.2.5} proves the ferromagnetic phase transition at $t_\cc$ (Proposition~\ref{p.mag}).
\end{itemize}
The developments in Sections~\ref{s.2.1} and~\ref{s.2.2} enable the proof of Proposition~\ref{p.t_c}, which in turn serves as a foundation for the separate proofs of Propositions~\ref{p.color_symmetry} and~\ref{p.mag}.
Theorem~\ref{thm.main} is established by combining Propositions~\ref{p.t_c}, \ref{p.color_symmetry}, and~\ref{p.mag}.
}

\subsection{Properties of the initial condition}\label{s.2.1}

We derive properties of $\sP_\beta$ given in~\eqref{e.sP(x)=}.

\begin{remark}[Matching notation in \cite{chen2023on}]\label{r.compare_note}

Since we will use results from~\cite{chen2023on}, 
we clarify that
\begin{align*}
    F_{\beta,N}(0,x) = F_N(0,\diag(x)),\quad\forall x \in \R^\D,
\end{align*}
in the Potts setting, where free energy on the right-hand side is the one considered in~\cite{chen2023on}. The relation holds due to~\eqref{e.self-overlap=m_N}. As a result of~\eqref{e.sP(x)=} and~\cite[(1.3) and (1.10)]{chen2023on}, we have
\begin{align*}
    \sP_\beta(x) = \sP(\diag(x)),\quad\forall x\in\R^\D
\end{align*}
where $\sP$ is defined in~\cite[(1.10)]{chen2023on}. There, the domain for $F_N$ and $\sP$ is the set of $\D\times\D$ real symmetric matrices. We also need one more important identity. Since the self-overlap $N^{-1}\sigma\sigma^\intercal$ is diagonal (see~\eqref{e.self-overlap=m_N}), we can infer from~\cite[Theorem~1.1 (1)]{chen2023on} that the symmetric matrix-valued $\nabla\sP$ is always diagonal. Using this and the previous display, we have
\begin{align}\label{e.diag(nabla_sP)}
    \diag\Ll(\nabla \sP_\beta(x)\Rr) = \nabla \sP(\diag(x)),\quad\forall x \in \R^\D.
\end{align}
Here, the everywhere differentiability of $\sP$ (and thus $\sP_\beta$) is ensured by~\cite[Proposition~2.3]{chen2023on}.
\qed
\end{remark}

\begin{lemma}[Properties of the initial condition]\label{l.sP_property}
The function $\sP_\beta:\R^\D\to\R$ is differentiable everywhere, Lipschitz, convex, and semi-concave. Moreover, $\one \cdot\nabla\sP_\beta(x) =1 $ for every $x\in \R^\D$.
\end{lemma}

Here, a function is semi-concave if, after subtracting some quadratic function, it becomes concave.

\begin{proof}
We use Remark~\ref{r.compare_note} and results from~\cite{chen2023on}.
We know from~\cite[Lemma~2.2]{chen2023on} that the first and second order derivatives of $\sP_\beta(\pi,\cdot)$ are bounded everywhere on $\R^\D$ and uniformly in $\pi$. 
It is well known (e.g.\ see \cite[Proposition~1.5]{cannarsa2004semiconcave}) that the infimum of functions with uniform bounds on concavity is semi-concave.
Hence, $\sP_\beta$ is semi-concave.
Other properties in the first sentence of the statement can be found in~\cite[Lemma~2.1 and Proposition~2.3]{chen2023on}.
To show the second part, we use Remark~\ref{r.compare_note} to see
\begin{align}\label{e.1cdotnablasP_beta}
    \one \cdot \nabla \sP_\beta(x) = \diag(\one)\cdot \nabla\sP(\diag(x)),\quad\forall x \in \R^\D.
\end{align}
Here, on the right-hand side, $\nabla\sP$ is the derivative with respect to the entry-wise inner product of $\D\times\D$ matrices (Frobenius inner product). So, $\nabla\sP$ is $\R^{\D\times\D}$-valued. The dot product on the right-hand side is this inner product. Recall that the differentiability of $\sP$ is ensured by~\cite[Proposition~2.3]{chen2023on}. Applying~\cite[Theorem~1.1~(1)]{chen2023on} to the right-hand side of~\eqref{e.1cdotnablasP_beta}, we have
\begin{align*}
    \one\cdot \nabla\sP_\beta(x) = \lim_{N\to\infty}\diag(\one) \cdot \E \la N^{-1}\sigma\sigma^\intercal  \ra
\end{align*}
where the Gibbs measure $\la\cdot\ra$ is equal to~ the one in \eqref{e.gibbs} at $t=0$ and this $x$. From~\eqref{e.self-overlap=m_N}, we can deduce that $\diag(\one)\cdot \Ll(N^{-1}\sigma\sigma^\intercal\Rr)=1$. Inserting this into the above display gives the second part of this lemma.
\end{proof}

\subsection{Limit of free energy}\label{s.2.2}

Notice that the convex conjugate of $|\cdot|^2$ is $\frac{1}{4}|\cdot|^2$.
The following formula in~\eqref{e.hopf-lax} is known as the Hopf--Lax formula.
\begin{proposition}[Hopf--Lax formula for the limit]\label{p.cvg_F_N}
As $N\to\infty$, $F_{\beta,N}$ converges pointwise on $\R_+\times \R^\D$ to a Lipschitz and convex function $f_\beta$ given by
\begin{align}\label{e.hopf-lax}
    f_\beta(t,x) = \sup_{y\in\R^\D} \Ll\{\sP_\beta(x+y)-\frac{|y|^2}{4t}\Rr\},\quad\forall (t,x)\in\R_+\times \R^\D.
\end{align}
Moreover, for every $(t,x)\in\R_+\times \R^\D$, maximizers of the right-hand side in~\eqref{e.hopf-lax} exist.
\end{proposition}

\begin{proof}
This proposition can be deduced following the proof of~\cite[Proposition~5.2]{chen2023self}. The difference is that, there, the domain for the second variable is the space of real symmetric matrices, while the domain here is $\R^\D$. The simplification here is allowed by the fact that the self-overlap $N^{-1}\sigma\sigma^\intercal$ is diagonal as in~\eqref{e.self-overlap=m_N}.
The convexity and Lipschitzness of $f_\beta$ are due to the fact that $F_{N,\beta}$ is convex and Lipschitz uniformly in $N$, which can be deduced as in~\cite[Lemma~5.1]{chen2023self}. Since $\sP_\beta$ is Lipschitz (see Lemma~\ref{l.sP_property}) and thus grows at most linearly, the existence of maximizers follows.
\end{proof}

We need the following lemma to characterize the differentiability of $f_\beta(t,\cdot)$.

\begin{lemma}[Envelop theorem]\label{l.envelop}
Let $x\in\R^\D$.
Then, $f_\beta(t,\cdot)$ is differentiable at $x$ if and only if the set
\begin{align*}
    \Ll\{\nabla \sP_\beta(x+y)\ \big| \ \text{$y$ maximizes the right-hand side of~\eqref{e.hopf-lax}}\Rr\}
\end{align*}
is a singleton. In this case, $\nabla f_\beta(t,x) = \nabla\sP_\beta(x+y)$ for every maximizer $y$.
\end{lemma}

\begin{proof}
This is a result of the envelope theorem, for which we refer to the version stated as~\cite[Theorem~2.21]{HJbook}. The linear growth of $\sP_\beta$ as implied by Lemma~\ref{l.sP_property} ensures the conditions in that theorem are satisfied.
\end{proof}

Recall the action of color permutation in~\eqref{e.x^s=} and~\eqref{e.permute_spin}.
Since the uniform distribution of Potts spins is invariant under the action as in~\eqref{e.permute_spin}, 
a simple observation from~\eqref{e.F_N=}, \eqref{e.sP(x)=}, and~\eqref{e.hopf-lax} is that
\begin{align}\label{e.symmetry_F_sP_f}
    F_{\beta,N}(t,x) = F_{\beta,N}(t,x^\bfs) ,\qquad \sP_\beta(x) = \sP_\beta(x^\bfs) ,\qquad f_\beta(t,x) = f_\beta(t,x^\bfs)
\end{align}
for every $(t,x) \in\R_+\times \R^\D$ and $\bfs \in \Sym$. As a result, we also have
\begin{align}\label{e.nabla_sP=D^-1one}
    x \in \R\one \quad \Longrightarrow \quad \nabla\sP_\beta(x) = \D^{-1} \one.
\end{align}
Here and henceforth, $\R\one = \{r\one:\:r\in\R\}$.
Indeed, from~\eqref{e.symmetry_F_sP_f}, we get $\nabla\sP_\beta(x^\bfs) = \Ll(\nabla\sP_\beta(x)\Rr)^\bfs$ for every $x\in\R^\D$. Then, if $x\in \R\one$, we have $x^\bfs =x$ and thus $\nabla\sP_\beta(x) = \Ll(\nabla\sP_\beta(x)\Rr)^\bfs$ for every $\bfs\in\Sym$. This along with $\one\cdot \nabla\sP(x)=1$ in Lemma~\ref{l.sP_property} verifies~\eqref{e.nabla_sP=D^-1one}.

We are mostly interested in the zero external field case. 
From Proposition~\ref{p.cvg_F_N}, we have, for every $t\in\R_+$,
\begin{align}\label{e.f(t,0)=sup}
    \lim_{N\to\infty} F_{\beta,N}(t,0) = f_\beta(t,0) = \sup_{y\in\R^\D} \Ll\{\sP_\beta(y)-\frac{|y|^2}{4t}\Rr\}.
\end{align}
The only subset of $\R^\D$ that is invariant under the action of $\Sym$ as in~\eqref{e.x^s=} is $\R\one $. We are particularly interested in the case where the supremum in~\eqref{e.f(t,0)=sup} is achieved on $\R\one$.

\begin{lemma}\label{l.color_sym_max_in_R1}
For every $t\in\R_+$, we have
\begin{align}\label{e.f>=}
    f_\beta(t,0) \geq \sup_{y\in\R\one} \Ll\{\sP_\beta(y)-\frac{|y|^2}{4t}\Rr\} =  \sP_\beta(0)+\frac{t }{\D},
\end{align}
where the supremum has a unique maximizer $y=\frac{2t}{\D}\one$.
\end{lemma}

\begin{proof}
The inequality follows from~\eqref{e.hopf-lax}.
Since $\one \cdot m_N = 1$ (see~\eqref{e.self-overlap=m_N}), using the definition of $F_{\beta,N}$ in~\eqref{e.F_N=}
and~\eqref{e.sP(x)=}, 
we have
\begin{gather}
    F_{\beta,N}(t,x+r\one) = F_{\beta,N}(t,x)+ r,\quad\forall x \in \R^\D ,\, r\in \R; \notag
\\
    \label{e.sP(x+r1)=sP(x)+r}
    \sP_\beta(x+r\one) = \sP_\beta(x)+r,\quad\forall x \in \R^\D ,\, r\in \R.
\end{gather}
Thus, the equality in~\eqref{e.f>=} follows from~\eqref{e.sP(x+r1)=sP(x)+r} at $x=0$ and simple computation.
\end{proof}

\subsection{Critical point}\label{s.2.3}
We define
\begin{align}\label{e.t_c=}
    t_\cc = \inf\Ll\{t>0:\:f_\beta(t,0)> \sP_\beta(0)+\frac{t}{\D}\Rr\}.
\end{align}
This is the critical value of $t$ that marks the onset of color symmetry breaking and the transition.

\begin{proposition}[Properties of $t_\cc$]\label{p.t_c}
The following holds:
\begin{enumerate}
    \item \label{i.p.t_c_1} We have $0<t_\cc<\infty$.
    \item \label{i.p.t_c_2} For every $t<t_\cc$, we have that $f_\beta(t,\cdot)$ is differentiable at $0$ (with $\nabla f_\beta(t,0)= \D^{-1}\one$) and that the Hopf--Lax formula of $f_\beta(t,0)$ as in~\eqref{e.hopf-lax} has a unique maximizer $\frac{2t}{\D}\one$.
    \item \label{i.p.t_c_3} For every $t>t_\cc$, we have that $f_\beta(t,0)>\sP_\beta(0)+\frac{t}{\D}$; that $f_\beta(t,\cdot)$ is not differentiable at $0$; and that the Hopf--Lax formula of $f_\beta(t,0)$ has more than one maximizer, none of which is equal to $\frac{2t}{\D}\one$.
\end{enumerate}
\end{proposition}

From this proposition, we can obtain different equivalent definitions of $t_\cc$: as the supremum of $t$ such that $f_\beta(t,0)=\sP(0)+\frac{t}{\D}$; or that $f_\beta(t,\cdot)$ is differentiable at $0$; or that the Hopf--Lax formula has a unique maximizer. 

In the remainder of this subsection, we prove this proposition part by part.

\begin{proof}[Proof of $t_\cc>0$ in Proposition~\ref{p.t_c}~\eqref{i.p.t_c_1}]
Since $\sP_\beta$ is semi-concave (see Lemma~\ref{l.sP_property}), we have that $y\mapsto \sP_\beta(y)-\frac{|y|^2}{4t}$ is strictly concave for sufficiently small $t$. Hence, the maximizer in~\eqref{e.hopf-lax} must be unique, which we denote by $y_0$. 
Due to the symmetry of $\sP_\beta$ as in~\eqref{e.symmetry_F_sP_f}, we can see that $y_0^\bfs$ (defined as in~\eqref{e.x^s=}) is also a maximizer for any $\bfs \in\Sym$. Hence, the uniqueness of maximizers implies $y_0\in \R\one$. Then, Lemma~\ref{l.color_sym_max_in_R1} gives $f_\beta(t,0) = \sP_\beta(0)+\frac{t}{\D}$ for all sufficiently small $t>0$, which ensures $t_\cc>0$.
\end{proof}

\begin{proof}[Proof Proposition~\ref{p.t_c}~\eqref{i.p.t_c_2}]
Let $t< t_\cc$. By definition, we have that $f_\beta(t,0)= \sP_\beta(0)+\frac{t}{\D}$. 
We set $y_0 = \frac{2t}{\D}\one$ and we know from Lemma~\ref{l.color_sym_max_in_R1} that $y_0$ is a maximizer of the Hopf--Lax formula~\eqref{e.hopf-lax}. 
Due to~\eqref{e.nabla_sP=D^-1one}, we have $\nabla\sP_\beta(y_0)=\D^{-1}\one$. Now, we argue by contradiction and assume that $f_\beta(t,\cdot)$ is not differentiable at $0$. By the envelope theorem as stated in Lemma~\ref{l.envelop}, there must be a maximizer $y$ of~\eqref{e.hopf-lax} such that
\begin{align}\label{e.nabla_sP(y)_neq}
    \nabla \sP_\beta(y) \neq \D^{-1}\one.
\end{align}
In particular, this along with~\eqref{e.nabla_sP=D^-1one} implies $y\not\in \R \one$. Since $y$ is a maximizer and $t< t_\cc$, we have
\begin{align}\label{e.sP(y)-|y|^2/4t=sP(0)+t/D}
    \sP_\beta(y) - \frac{|y|^2}{4t} = \sP_\beta(0)+ \frac{t}{\D}.
\end{align}
In the following we define functions $g,\, g_0 : \R_{++}\to\R$ by
\begin{align*}
    g(s) = \sP_\beta(y)-\frac{|y|^2}{4s}, \qquad g_0(s) = \sP_\beta(0)+\frac{s}{\D},\qquad\forall s\in\R_{++}.
\end{align*}
We then analyze the behavior of these two functions as $s$ varies.

Since $y$ is a maximizer of~\eqref{e.hopf-lax}, we get
\begin{align}\label{e.nabla_sP(y)=y/2t}
    \nabla\sP_\beta(y) =\frac{y}{2t}.
\end{align}
As $\one\cdot\nabla\sP_\beta (y)=1$ due to Lemma~\ref{l.sP_property}, we have $\sum_{i=1}^\D \frac{y_i}{2t}=1$, which along with Jensen's inequality implies
\begin{align}\label{e.|y|^2/4t^2>1/D}
    \frac{|y|^2}{4t^2}\geq \frac{1}{\D}
\end{align}
where the equality is achieved if and only if $\frac{y}{2t}= \frac{1}{\D}\one$. This equality does not hold in this case due to~\eqref{e.nabla_sP(y)_neq} and~\eqref{e.nabla_sP(y)=y/2t}. Hence, the inequality in~\eqref{e.|y|^2/4t^2>1/D} is strict, which implies $g'(t)>g'_0(t)$. Meanwhile, \eqref{e.sP(y)-|y|^2/4t=sP(0)+t/D} gives $g(t)=g_0(t)$. Therefore, for all sufficiently small $\eps>0$, we have $g(t+\eps) > g_0(t+\eps)$, which along with the Hopf--Lax formula~\eqref{e.hopf-lax} yields
\begin{align*}
    f_\beta(t+\eps,0) > \sP_\beta(0)+\frac{t+\eps}{\D}.
\end{align*}
We can choose $\eps$ to satisfy $t+\eps<t_\cc$, which then reaches a contradiction with the definition of $t_\cc$.
Hence, $f_\beta(t,\cdot)$ must be differentiable at $0$ by the argument of contradiction.

Let $y$ be any maximizer of the Hopf--Lax formula. The envelope theorem stated as in Lemma~\ref{l.envelop} gives $\nabla f_\beta(t,0)= \nabla\sP_\beta(y)$. 
By~\eqref{e.symmetry_F_sP_f}, we also have $\nabla f_\beta(t,0) = \Ll(\nabla f_\beta(t,0)\Rr)^\bfs$ for every $\bfs \in \Sym$, which along with $\one\cdot \nabla f_\beta(t,0) = \one\cdot \nabla \sP_\beta(y) = 1$ (due to Lemma~\ref{l.sP_property}) implies $\nabla f_\beta(t,0) = \D^{-1} \one$.
Hence, this along with~\eqref{e.nabla_sP(y)=y/2t} yields $\frac{y}{2t} = \nabla f_\beta(t,0) = \D^{-1}\one$, which completes proof. 
\end{proof}

\begin{proof}[Proof Proposition~\ref{p.t_c}~\eqref{i.p.t_c_3}]
Notice that $s\mapsto f_\beta(s,0)$ is convex (see~\eqref{e.hopf-lax}), $f_\beta(0,0) =\sP_\beta(0)$, and $f_\beta(s,0)\geq \sP_\beta(0)+\frac{s}{\D}$ for every $s$. Hence, if $f_\beta(t,0)=\sP_\beta(0)+\frac{t}{\D}$ at some $t>t_\cc$, then we have $f_\beta(s,0)=\sP_\beta(0)+\frac{s}{\D}$ for every $s\in[t_\cc,t]$, which contradicts the definition of $t_\cc$. Therefore, we have $f_\beta(t,0)>\sP_\beta(0)+\frac{t}{\D}$ for every $t>t_\cc$. 

For the second statement, suppose that $f_\beta(t,\cdot)$ is differentiable at $0$ at some $t>t_\cc$. Let $y$ be any maximizer of the Hopf--Lax formula~\eqref{e.hopf-lax}. Then, from the last paragraph in the proof of Part~\eqref{i.p.t_c_2}, we have $y=2t \D^{-1}\one$. This implies that the Hopf--Lax formula~\eqref{e.hopf-lax} of $f_\beta(t,\cdot)$ maximizes over $\R\one$ which by Lemma~\ref{l.color_sym_max_in_R1} gives $f_\beta(t,0) = \sP_\beta(0)+\frac{t}{\D}$. This contradicts the result in the previous paragraph. Hence, $f_\beta(t,\cdot)$ cannot be differentiable at $0$ for any $t>t_\cc$.

Lastly, the Hopf--Lax formula must have more than one maximizer because otherwise, we can use the envelope theorem as stated in Lemma~\ref{l.envelop} to deduce that $f_\beta(t,\cdot)$ is differentiable at $0$. Also, none of the maximizers can be $\frac{2t}{\D}\one$, because otherwise we can use Lemma~\ref{l.color_sym_max_in_R1} to deduce $f_\beta(t,0)=\sP_\beta(0)+\frac{t}{\D}$, which contradicts the first part of the statement.
\end{proof}

To show that $t_\cc<\infty$, we need some preparation.
Denote by $F_N^\cw(t,x)$ be the free energy in~\eqref{e.F_N=} with $\beta$ set to be zero.
This is the free energy in the Curie--Weiss--Potts model.
We denote by $\psi = F_1^\cw(0,\cdot)$ the initial condition. In this case, we have $\psi = F_N^\cw(0,\cdot)$ for every $N\in\N$ and the explicit expression
\begin{align}\label{e.psi=_2}
    \psi(x) =\log \Ll(\sum_{d=1}^\D e^{x_d}\Rr)-\log \D,\quad\forall x\in\R^\D.
\end{align}
Proposition~\ref{p.cvg_F_N} yields that, for every $(t,x)\in\R_+\times \R^\D$,
\begin{align}\label{e.limF^CW_N_2}
    \lim_{N\in\infty}F^\cw_N(t,x) = \sup_{y\in\R^\D}\Ll\{\psi(x+y) - \frac{|y|^2}{4t}\Rr\}.
\end{align}
The difference between the free energy of the Potts spin glass and that of Curie--Weiss--Potts is also bounded.
\begin{lemma}\label{l.F-F^CW_2}
We have
\begin{align*}
    F_N^\cw(t,x) -\frac{\beta^2}{2}\leq F_{\beta,N}(t,x) \leq F_N^\cw(t,x)
\end{align*}
uniformly in $(t,x)\in\R_+\times \R^\D$ and $N\in\N$.
\end{lemma}

\begin{proof}
For $s\in[0,1]$, we define $\phi(s) =  F_{\beta\sqrt{s},N}(t,x)$. Then, we can compute the derivative $\phi'(s)$ of $\phi(s)$. By using the Gaussian integration by parts and the boundedness of Potts spins, we can show $-\frac{\beta^2}{2}\leq \phi'(s)\leq 0$. Since $\phi(1) = F_{\beta,N}(t,x)$ and $\phi(0) = F_N^\cw(t,x)$, we can get the desired result.
\end{proof}

\begin{proof}[Proof of $t_\cc<\infty$ in Proposition~\ref{p.t_c}~\eqref{i.p.t_c_1}]
For any $t>0$, set $y(t) = (2t,0,\cdots,0)\in\R^\D$. Using~\eqref{e.psi=_2}, we can compute
\begin{align}\label{e.psi(y(t))-|y(t)|^2/rt}
    \psi(y(t))-\frac{|y(t)|^2}{4t}= \log (e^{2t}+\D-1)-\log \D - t \geq t- \log \D
\end{align}
Then, we have
\begin{align*}
    f_\beta(t,0)=\lim_N F_{\beta,N}(t,0) \stackrel{\text{L.\ref{l.F-F^CW_2}}}{\geq} \lim_N F^\cw_N(t,0) - \frac{\beta^2}{2} \stackrel{\eqref{e.limF^CW_N_2},\eqref{e.psi(y(t))-|y(t)|^2/rt}}{\geq} t- \log \D - \frac{\beta^2}{2}
\end{align*}
which implies that, for all $t$ sufficiently large,
\begin{align*}
    f_\beta(t,0) > \sP_\beta(0) + \frac{t}{\D}
\end{align*}
and thus $t_\cc<\infty$ by its definition in~\eqref{e.t_c=}.
\end{proof}

\subsection{Color symmetry breaking}\label{s.2.4.color_sym_break}

We show that $t_\cc$ in~\eqref{e.t_c=} marks the color symmetry breaking.

\begin{proposition}[Color symmetry]\label{p.color_symmetry}
The following holds:
\begin{enumerate}
    \item If $t\in [0,t_\cc]$, then the color symmetry is preserved and we have
    \begin{align}\label{e.color_sym_f=}
        f_\beta(t,0) = \inf_{\zeta\in\mathcal{Z}}\sP_\beta(\Psi\circ\zeta, 0) + \frac{t}{\D}.
    \end{align}
    \item If $t>t_\cc$, then the color symmetry is broken and we have
    \begin{align*}
        f_\beta(t,0) >\inf_{\zeta\in\mathcal{Z}}\sP_\beta(\Psi\circ\zeta, 0) + \frac{t}{\D}.
    \end{align*}
\end{enumerate}
\end{proposition}

We remark that \cite[Theorem~1.2]{chen2023onparisi} gives a rewriting of the infimum in~\eqref{e.color_sym_f=} which admits a unique minimizer. This result is based on an extension of the arguments in~\cite{aufche} (also see~\cite{chen2023parisi}).

\begin{proof}
To show the first part, let $t\in[0,t_\cc]$. 
By the definition of $t_\cc$ in~\eqref{e.t_c=} and continuity (needed for $t=t_\cc$), we have $f_\beta(t,0)=\sP_\beta(0)+\frac{t}{\D}$. Lemma~\ref{l.color_sym_max_in_R1} implies that $y=\frac{2t}{\D}\one$ is a maximizer of the Hopf--Lax formula in~\eqref{e.hopf-lax} and thus also of the supremum in~\eqref{e.f(t,0)=hopf_lax_path}. By~\eqref{e.nabla_sP=D^-1one}, we have $\nabla\sP_\beta(y) =\D^{-1}\one$. By~\cite[Theorem~1.1]{chen2023onparisi}, we have
\begin{align}\label{e.sP(0)=inf_zeta}
    \sP_\beta(0) = \inf_{\zeta\in\mathcal{Z}}\sP_\beta(\Psi\circ\zeta, 0)
\end{align}
and thus
\begin{align*}
    f_\beta(t,0) = \inf_{\zeta\in\mathcal{Z}}\sP_\beta(\Psi\circ\zeta, 0) + \frac{t}{\D} \stackrel{\eqref{e.sP(x+r1)=sP(x)+r}}{=}\inf_{\zeta \in \mathcal Z}\sP_\beta(\Psi\circ \zeta,y) - \frac{|y|^2}{4t}
\end{align*}
which proves the first part.

Now, let $t>t_\cc$. Suppose that the color symmetry is preserved. Then, there is a maximizer $y$ of the Hopf--Lax formula~\eqref{e.hopf-lax} such that $\nabla\sP_\beta(y) =\D^{-1}\one$. The maximality implies that $\nabla\sP_\beta(y) = \frac{y}{2t}$ and thus $y= \frac{2t}{\D}\one$, which contradicts Proposition~\ref{p.t_c}~\eqref{i.p.t_c_3}. Hence, the color symmetry is broken. Proposition~\ref{p.t_c}~\eqref{i.p.t_c_3} also gives $f_\beta(t,0)>\sP_\beta(0)+\frac{t}{\D}$ which together with~\eqref{e.sP(0)=inf_zeta} gives the desired inequality.
\end{proof}

\subsection{Ferromagnetism}\label{s.2.5}

We describe the phase transition at $t_\cc$ in terms of the change of magnetization. 

\begin{proposition}[Magnetization]\label{p.mag}
The following holds.
\begin{enumerate}
    \item \label{i.p.mag_1} (``Zero'' magnetization in the sub-critical regime) Let $t\in[0,t_\cc)$ and let $(x_N)_{N\in\N}$ be any sequence converging to $0$. Then,
    \begin{align*}
        \lim_{N\to\infty} \E \la \Ll|m_N-\D^{-1}\one\Rr|\ra_{t,x_N} =0.
    \end{align*}

    \item \label{i.p.mag_2} (Spontaneous magnetization in the super-critical regime)
    Let $t>t_\cc$. Then, there is a strictly increasing sequence $(N_n)_{n\in\N}$ of positive integers and $(x_n)_{n\in\N}$ converging to $0$ such that $\lim_{n\to\infty} \E \la m_{N_n}\ra_{t,x_n}$ exists but is not $\D^{-1}\one$.
\end{enumerate}
\end{proposition}

\rv{In the first part of Proposition~\ref{p.color_symmetry}, the identity~\eqref{e.color_sym_f=} holds at $t_\cc$, as a simple consequence of continuity. However, due to the lack of understanding of the differentiability of $f_{\beta}(t_\cc,\cdot)$ at $0$, results such as Proposition~\ref{p.mag} fail to provide information at the critical point $t = t_\cc$.}

\begin{proof}[Proof of Proposition~\ref{p.mag}~\eqref{i.p.mag_1}]
By computing the first-order derivative of $F_{\beta,N}(t,\cdot)$ using the expression in~\eqref{e.F_N=}, we have
\begin{align}\label{e.nabla_F_N=E<m>}
    \nabla F_{\beta,N}(t,x) = \E \la m_N\ra_{t,x},\quad\forall x\in\R^\D.
\end{align}
In particular, this implies that $F_{\beta,N}(t,\cdot)$ is Lipschitz uniformly in $N$. We can use this to upgrade the pointwise convergence in Proposition~\ref{p.cvg_F_N} to that $F_{\beta,N}(t,\cdot)$ converges to $f_\beta(t,\cdot)$ uniformly on every compact set.
By computing the second-order derivatives of $F_{\beta,N}(t,\cdot)$ (similar to~\cite[(2.4)]{chen2023on}), we can verify that $F_{\beta,N}(t,\cdot)$ is convex. Then, for any $y\in\R^\D$ and $r>0$, the convexity gives
\begin{align*}
    \frac{F_{\beta,N}(t,x_N) - F_{\beta,N}(t,x_N-ry)}{r} \leq y\cdot \nabla F_{\beta,N}(t,x_N)\leq \frac{F_{\beta,N}(t,x_N+ry) - F_{\beta,N}(t,x_N)}{r}.
\end{align*}
First sending $N\to\infty$ and then $r\to0$, we can use the local uniform convergence of $F_{\beta,N}$ and the differentiability of $f_\beta(t,\cdot)$ at $0$ to get $\lim_{N\to\infty} y\cdot \nabla F_{\beta,N}(t,x_N) = y\cdot \nabla f_\beta(t,0)$. Using this along with~\eqref{e.F_N=} and $\nabla f_\beta(t,0) =\D^{-1}\one$ (Proposition~\ref{p.t_c}~\eqref{i.p.t_c_2}), we get
\begin{align*}
    \lim_{N\to\infty} \E \la m_N\ra_{t,x_N}= \D^{-1}\one.
\end{align*}
The proof will be complete if we can show
\begin{align*}
    \lim_{N\to\infty} \E \la \Ll|m_N-\E \la m_N\ra_{t,x_N}\Rr|\ra_{t,x_N} =0.
\end{align*}
This is a standard consequence of the concentration of free energy, the convexity of $F_{\beta,N}(t,\cdot)$ (and its un-averaged version), and the differentiability of $f_\beta(t,\cdot)$ at $0$ (due to $t<t_\cc$). For instance, one can use the same argument for~\cite[Proposition~2.4]{chen2023on} verbatim by substituting $m_N$, $x_N$, $F_{\beta,N}(t,\cdot)$, $f_\beta(t,\cdot)$ for $N^{-1}\sigma\sigma^\intercal$, $x$, $F_N$, $\sP$ therein. The detail is tedious and thus omitted here.
\end{proof}

\begin{proof}[Proof of Proposition~\ref{p.mag}~\eqref{i.p.mag_2}]
Recall that $f_\beta(t,\cdot)$ is convex (see Proposition~\ref{p.cvg_F_N}). If for every sequence $(x_n)_{n\in\N}$ of differentiable points of $f_\beta(t,\cdot)$ converging to $0$ we have $\lim_{n\to\infty} \nabla f_\beta(t,x_n) = \D^{-1}\one$, then we can use the characterization of differentials of convex function as stated in~\cite[Theorem~25.6]{rockafellar1970convex} to deduce that $f_\beta(t,\cdot)$ is differentiable at $0$. Since we know that $f_\beta(t,\cdot)$ is not differentiable at $0$ due to Proposition~\ref{p.t_c}~\eqref{i.p.t_c_3}, there must be $(x_n)_{n\in\N}$ such that $(\nabla f_\beta(t,x_n))_{n\in\N}$ does not converge to $\D^{-1}\one$. 
Since $f_\beta(t,\cdot)$ is Lipschitz (see Proposition~\ref{p.cvg_F_N}), the derivatives of $f_\beta(t,\cdot)$ are bounded, which allows us to pass to a subsequence (still denoted as $(x_n)_{n\in\N}$ for convenience) along which $\nabla f_\beta(t,x_n)$ converges to some $a\in\R^\D$ other than $\D^{-1}\one$. 
As argued in the previous proof, we have $\lim_{N\to\infty} \nabla F_{\beta,N}(t,x_n) = \nabla f_\beta(t,x_n)$ for each $n$.
Then, we can find a strictly increasing sequence $(N_n)_{n\in\N}$ such that $\nabla F_{\beta,N_n}(t,x_n)$ approximates $\nabla f_\beta(t,x_n)$ and thus $\lim_{n\to\infty} \nabla F_{\beta,N_n}(t,x_n) = a$. This along with~\eqref{e.nabla_F_N=E<m>} gives the result.
\end{proof}

\noindent
\textbf{Acknowledgments.}
The author thanks Victor Issa and Jean-Christophe Mourrat for helpful discussions. 
The author thanks the anonymous referee for valuable suggestions that improved the manuscript.

\noindent
\textbf{Funding.}
The author is funded by the Simons Foundation.

\noindent
\textbf{Data availability.}
No datasets were generated during this work.

\noindent
\textbf{Conflict of interests.}
The author has no conflicts of interest to declare.

\noindent
\textbf{Competing interests.}
The author has no competing interests to declare.

\small
\bibliographystyle{abbrv}
\newcommand{\noop}[1]{} \def\cprime{$'$}

\end{document}